\documentclass[12pt,a4paper,headings=small,oneside]{scrartcl}

\usepackage{amsmath, amssymb, amsfonts}
\usepackage{amsthm}
\usepackage{mathtools}
\usepackage{graphics}
\usepackage{epstopdf}
\usepackage{natbib}
\usepackage{bbm}
\usepackage{bm}
\usepackage{setspace}
\usepackage{color}
\usepackage{enumitem}
\setenumerate[0]{label=\alph*)} 
\usepackage[hyphens]{url}
\usepackage{hyperref}

\newtheorem{theorem}{Theorem}

\newtheorem{property}[theorem]{Property}

\newtheorem{corollary}[theorem]{Corollary}

\newtheorem{proposition}[theorem]{Proposition}

\theoremstyle{definition}

\newtheorem{example}[theorem]{Example}
\newtheorem{definition}[theorem]{Definition}
\newtheorem{remark}[theorem]{Remark}

\def\R{\mathbb{R}}
\def\Z{\mathbb{Z}}
\def\N{\mathbb{N}}
\def\Nz{\N_{\geq 2}}	

\def\Ddists{\mathcal{D}}
\def\PDdists{\mathcal{D}_0}
\def\supp{\operatorname{supp}}
\def\indsetset{\mathbb{I}}
\def\idnseqset{\mathbb{S}}
\def\Pr{\operatorname{Pr}}
\newcommand{\smi}[1]{\underline{#1}}	

\newcommand{\discoa}{\leq_{\text{disp}}^{\land\text{-disc}}}	
\newcommand{\stdCompj}{\leftrightharpoons}	
\newcommand{\stdCompja}{\leftrightharpoons_\land}	
\newcommand{\cdfis}{\; \hat{=} \;}		
\newcommand{\nn}[3]{\operatorname{NN}_{#2}^{#3}(#1)}	

\allowdisplaybreaks

\usepackage[colorinlistoftodos,textsize=tiny]{todonotes}
\newcommand{\Comments}{1}
\newcommand{\mytodo}[2]{\ifnum\Comments=1%
	\todo[linecolor=#1!80!black,backgroundcolor=#1,bordercolor=#1!80!black]{#2}\fi}

\ifnum\Comments=1               
\fi

\title{An easily verifiable dispersion order for discrete distributions}

\author{Andreas Eberl$^{1}$, Bernhard Klar$^{2}$, and Alfonso Su\'arez-Llorens$^{3}$ \\
  \small{$^{1}$ Institute of Statistics, Karlsruhe Institute of Technology (KIT), Germany,  
  \footnote{andreas.eberl@kit.edu}} \\
  \small{$^{2}$ Institute of Stochastics, Karlsruhe Institute of Technology (KIT), Germany, \footnote{bernhard.klar@kit.edu}} \\
  \small{$^{3}$ Facultad de Ciencias, Dpto. Estadística e IO., Universidad de Cádiz, Cádiz,
Spain.\footnote{alfonso.suarez@uca.es}}
}
\date{\today}

\begin{document}
\maketitle

\begin{abstract}
Dispersion is a fundamental concept in statistics, yet standard approaches - especially via stochastic orders - face limitations in the discrete setting. In particular, the classical dispersive order, well-established for continuous distributions, becomes overly restrictive for discrete random variables due to support inclusion requirements. To address this, we propose a novel weak dispersive order for discrete distributions. This order retains desirable properties while relaxing structural constraints, thereby broadening applicability. We further introduce a class of variability measures based on probability concentration, offering robust and interpretable alternatives that conform to classical axioms. Empirical illustrations highlight the practical relevance of this framework.
\end{abstract}


\section{Introduction}\label{Sec1} 

The concept of dispersion - or variability - is a cornerstone of statistics, playing a pivotal role in both theoretical developments and applied methodologies. However, when we speak of ``dispersion'', it is essential to clarify what precisely is being measured and how this measurement is defined. In this context, two foundational questions arise: Given a random variable $X$ and a functional $\nu(X)$, under what conditions can $\nu(X)$ be considered a valid measure of dispersion? Furthermore, given two random variables $X$ and $Y$, when can we say that $X$ is more dispersed than $Y$?

The first question was systematically addressed by \cite{BiLe-1976,BiLe-1979}, who identified a set of core properties that any dispersion functional should satisfy. Specifically, a measure of dispersion $\nu(X)$ should be non-negative, i.e., $\nu(X) \geq 0$, with equality if $X$ is degenerate.
It should be translation invariant, so that $\nu(X + k) = \nu(X)$ for any constant $k$, and it should satisfy absolute homogeneity: $\nu(\lambda X) = |\lambda|\, \nu(X)$. A classical example fulfilling these conditions is the standard deviation. In some contexts, absolute homogeneity is relaxed to positive homogeneity, allowing rescaling only by positive factors.

The second question -- how to compare the dispersion of two random variables -- has been addressed through stochastic ordering. Several types of stochastic orders have been introduced for this purpose (see, e.g., \cite{shaked-sh-2006,mueller-stoyan}). Among these, the \emph{dispersive order}, denoted $\leq_{\text{disp}}$, introduced by Bickel and Lehmann and further developed by \cite{LeTo-1981}, plays a central role. According to this order, a random variable $X$ is more dispersed than $Y$ if the differences between corresponding quantiles of $X$ consistently exceed those of $Y$. This order is particularly strong, implying a range of dispersion-related comparisons and preserving important measures such as the standard deviation. For these reasons, \cite{BiLe-1979} and \cite{oja} viewed the dispersive order as a canonical tool for formalizing dispersion comparisons.

In the continuous case, the dispersive order has proven to be a powerful framework for comparing well-known distribution families and has been shown to preserve various classical dispersion measures, including the standard deviation, the Gini mean difference, differential entropy, and the $L_2$-norm of the density. These properties have made it valuable in fields such as finance, actuarial science, and reliability theory.

Although stochastic orderings for discrete distributions have generally received considerable attention in the literature (see, e.g.,  \cite{giowynn:2008, klenmat:2010}), the application of the dispersive order to such distributions presents serious limitations. As discussed in \cite{ek-2024} and formalized in Theorem 1.7.3 of \cite{mueller-stoyan}, a necessary condition for $F \leq_{\text{disp}} G$ is that the support of $F$ must be contained within that of $G$. While this condition is benign for continuous distributions, it becomes prohibitively restrictive in the discrete case. It effectively excludes most lattice distributions and empirical distributions with ties or differing sample sizes. These challenges stem from the discontinuities in the cumulative distribution function and the irregular spacing of the quantile levels, which complicate quantile-based comparisons in the discrete domain.

To address these difficulties, Eberl and Klar proposed a new dispersive order tailored to univariate discrete distributions. This refined order preserves many appealing features of the continuous dispersive order while adapting them to the discrete context. However, despite this advancement, the Eberl-Klar order remains relatively strong and may not accommodate all practical needs.

The goal of this paper is to go further by introducing a more flexible and weaker form of dispersion comparison for discrete random variables. Our approach builds on the notion of probability concentration originally formulated by \cite{levy37}, and we establish connections to classical concepts such as majorization, entropy, and randomness. Through a detailed study of the properties of this proposed order, we underscore the necessity for dispersion measures that are well-suited to the structure of discrete data.

This line of reasoning also motivates the introduction of a new family of dispersion measures specifically designed for discrete random variables. These measures fulfill the axiomatic framework of Bickel and Lehmann while offering the adaptability required to capture the 	peculiarities of discrete distributions. Altogether, our contribution offers a unified and interpretable framework for analyzing dispersion in discrete settings, enriching and extending existing methodologies.

The structure of the paper is as follows. In the next section, we introduce the weak dispersive order for discrete distributions. Section \ref{sec3} examines its main properties, while Section \ref{sec:disc-disp} demonstrates that the proposed order is strictly weaker than the discrete dispersive order of Eberl and Klar. In Section \ref{sec5}, we define new measures of variability based on the concentration function and analyze their properties. Finally, several empirical illustrations and concluding remarks are presented in Sections \ref{sec6} and \ref{sec7}.

\section{The weak dispersive order for discrete distributions} \label{sec2}

In \cite{fersur2003}, the authors introduced the weak dispersive order specifically for continuous distributions, providing several characterizations and analyzing its relationship with both the classical dispersive order and the majorization order. We recall its definition in what follows.

Let $X$ and $Y$ be two random variables with continuous distribution functions $F$ and $G$, respectively. The random variable $X$ is said to be less weakly dispersive than $Y$, denoted $X \leq_{\text{wd}} Y$ or equivalently $F \leq_{\text{wd}} G$, if for all $\varepsilon > 0$ the following inequality holds:
\[
\sup_{x_0} [F(x_0 + \varepsilon) - F(x_0)] \geq \sup_{y_0} [G(y_0 + \varepsilon) - G(y_0)].
\]
The authors also say that $X$ and $Y$ are equally dispersed in the weakly dispersive sense, denoted by $X =_{wd} Y$, if both $X \leq_{\text{wd}} Y$ and $Y \leq_{\text{wd}} X$ hold. It is easily seen that the relation $\leq_{\text{wd}}$ defines a partial order on the set of continuous distribution functions of real-valued random variables.

As mentioned above, \cite{fersur2003} considered only continuous distribution functions. In the continuous case, it is well known that
\[
\sup_{x_0} [F(x_0 + \varepsilon)-F(x_0)] = 
Q_X(\varepsilon) = \sup_{x_0} \Pr\{ X \in [x_0, x_0+\varepsilon] \},
\]
where $Q_X(\varepsilon)$ is the L\'evy concentration function, a widely used measure of probability concentration in the literature.

\begin{remark} 
	\begin{enumerate}
		\item[a)] Generally speaking, if $X \leq_{\text{wd}} Y$ holds, then whenever there exists an interval of length $\varepsilon$ in the support of $Y$, there also exists an interval of the same length in the support of $X$ that accumulates at least as much probability. The supremum is well-defined, since the expression $F(x + \varepsilon) - F(x)$ is bounded.
		\item[b)] 
		Intuitively, this reflects a different perspective compared to the classical dispersive order. In the classical setting, we fix a probability level $q - p$ and search for an interval that is more widely separated. Here, we fix the interval length $\varepsilon$ and look for an interval that accumulates more probability.
		\item[c)] 
		In the general case (i.e., not necessarily continuous distributions), the supremum may not be attained at a point in the support due to the right-continuity of $F$. However, this technicality is not crucial for the interpretation or application of the order.
	\end{enumerate}
\end{remark}

In the discrete case, a natural way to define the weak dispersive order is based directly on the L\'evy concentration function, as follows:

\begin{definition} \label{def2}
	Let $X$ and $Y$ be two discrete random variables with distribution functions $F$ and $G$, respectively. Then $X$ is said to be less weakly dispersive than $Y$, denoted $X \leq_{\text{wd}} Y$, or equivalently $F \leq_{\text{wd}} G$, if for all $\varepsilon > 0$, it holds that
	\begin{align*}
		Q_{X}(\varepsilon) = \sup_{x_0} \Pr\{ X \in [x_0, x_0 + \varepsilon] \} 
		\ \geq\ 
		Q_{Y}(\varepsilon) = \sup_{y_0} \Pr\{ Y \in [y_0, y_0 + \varepsilon] \}.
	\end{align*}
\end{definition}

It is well known that the concentration function $Q_X(\varepsilon)$ is right-continuous and non-decreasing in $\varepsilon$. Moreover, it is easy to verify that $Q_X(\varepsilon) \to 1$ as $\varepsilon \to \infty$. Finally, $Q_X(0)$ can be defined as the limit from the right:  
\[
Q_X(0) := \lim_{\varepsilon \to 0^+} Q_X(\varepsilon) = \sup \{ p_i \},
\]  
where $p_i = \Pr[X = x_i], i=1,2,\ldots,$ denotes the point masses of $X$.

Definition \ref{def2} can also be expressed in a more constructive form. Let $X$ and $Y$ be two discrete random variables with supports $R_X$ and $R_Y$, and associated probability masses $p_i = \Pr[X = x_i]$ and $q_i = \Pr[Y = y_i]$, respectively. The supports may be finite or countably infinite, and they may differ both in cardinality and in the location of their support points.

Given two support points $x_i, x_j \in R_X$, denote their distance by $d_{i,j}^X = |x_j - x_i|$. Analogously, define $d_{i,j}^Y = |y_j - y_i|$ for $y_i, y_j \in R_Y$.
Then, $X \leq_{\text{wd}} Y$ if and only if, for every $\varepsilon > 0$,
\begin{align} \label{oper}
	\max_{\substack{i,j: d_{i,j}^X \leq \varepsilon}} \sum_{k: x_i \leq x_k \leq x_j} p_k
	\ \geq\ 
	\max_{\substack{i,j: d_{i,j}^Y \leq \varepsilon}} \sum_{k: y_i \leq y_k \leq y_j} q_k.
\end{align}

For two lattice distributions on $\mathbb{N}_0$, we have $X \leq_{\text{wd}} Y$ if and only if, for all $m \in \mathbb{N}_0$,
\begin{align} \label{lat1}
	\sup_{i \in \mathbb{N}_0} \left\{ \sum_{k=i}^{i+m} p_k \right\} \geq 
	\sup_{j \in \mathbb{N}_0} \left\{ \sum_{k=j}^{j+m} q_k \right\}.
\end{align}

For any $i_0$ such that $\sum_{k=i_0}^{\infty} p_k < \max_k \{p_k\}$, we can replace the supremum over $i \in \mathbb{N}_0$ by a maximum over the finite set $\{0, \ldots, i_0\}$. The same holds for $Y$ with a corresponding index $j_0$. Setting $l_0 = \max\{i_0, j_0\}$, the condition in \eqref{lat1} becomes equivalent to
\begin{align} \label{lat2}
	\max_{i \in \{0, \ldots, l_0\}} \left\{ \sum_{k=i}^{i+m} p_k \right\} \geq 
	\max_{j \in \{0, \ldots, l_0\}} \left\{ \sum_{k=j}^{j+m} q_k \right\}.
\end{align}

Hence, the supremum in \eqref{lat1} is actually attained. The same argument applies to distributions on $\mathbb{Z}$ and to lattice distributions with arbitrary step size.

\section{Properties of the weak dispersive order} \label{sec3}

This section outlines several properties of the weak dispersive order. We begin with two elementary observations:

\begin{remark}
	\begin{enumerate}
		\item[a)] Suppose there exists a fixed $\varepsilon > 0$ that is smaller than the distance between any two distinct points in the support of both $X$ and $Y$; that is, $\varepsilon < d_{i,j}^X$ and $\varepsilon < d_{l,r}^Y$ for all $i < j$ and $l < r$. In this case, the condition $X \leq_{\text{wd}} Y$ implies that
		$\max \{p_i\} \geq \max \{q_i\}$
		This situation always occurs when both $X$ and $Y$ have finite supports.
		\item[b)] 
		The supremum defining the Lévy concentration function may be attained at multiple intervals. Since the set of rational numbers is dense in $\mathbb{R}$, it suffices to evaluate the weak dispersive order on a countable collection of $\varepsilon$ values. That is, the order relation can be verified by considering only a countable subset of positive $\varepsilon$.
	\end{enumerate}
\end{remark}

We next present a consequence of the weak dispersive order in terms of the range of the support:

\begin{property} \label{prop1}
	Let $X$ and $Y$ be discrete random variables with bounded supports. Denote the minimum and maximum elements of the support of $X$ by $x_l$ and $x_u$, respectively, and those of $Y$ by $y_l$ and $y_u$. Then, if $X \leq_{\text{wd}} Y$, it follows that
	\[
	x_u - x_l \leq y_u - y_l.
	\]
\end{property}

\begin{proof}
	The proof follows by contradiction, setting $\varepsilon = y_u - y_l$.
\end{proof}

The following result allows many comparisons. First, we recall the definition of a contraction function.
\begin{definition}
	Let $\phi: \mathbb{R} \mapsto \mathbb{R}$ be a real function. We will say that $\phi$ is a contraction if 
	$|\phi(y) -\phi(x)| \leq |y-x|$, for all $x, y\in \mathbb{R}$.
\end{definition}

\begin{theorem} \label{contrac}
	Let $X$ and $Y$ be two random variables such that $X=_{st} \phi(Y)$, where $\phi$ is a monotone contraction function. Then $X\leq_{\text{wd}}Y$. 
\end{theorem}

\begin{proof}
	Let us suppose first that $\phi$ is a non-decreasing contraction. Then
	\begin{eqnarray*} 
		Pr \left \{ Y \in \lbrack x_{0}, x_{0}+\varepsilon \rbrack \right \} & \leq & Pr \left \{ \phi(Y) \in \lbrack \phi(x_{0}), \phi(x_{0}+\varepsilon) \rbrack \right \}  \\
		& = &  Pr \left \{ X \in \lbrack \phi(x_{0}), \phi(x_{0}+\varepsilon) \rbrack \right \} \\
		& \leq & Pr \left \{ X \in \lbrack \phi(x_{0}), \phi(x_{0})+\varepsilon \rbrack \right \}.
	\end{eqnarray*}
	where the first inequality follows by just taking in account that $\phi$ is non-decreasing,  and the second equality by considering that $X=_{st} \phi(Y)$. For the last inequality, we use both assumptions $\phi$ is a non-decreasing function and $\phi$ is a contraction fuction. Then  $|\phi(x_0+\varepsilon) - \phi(x_0)| = \phi(x_0+\varepsilon) - \phi(x_0) \leq \varepsilon$, which leads that  $\lbrack \phi(x_{0}), \phi(x_{0}+\varepsilon) \rbrack \subseteq \lbrack \phi(x_{0}), \phi(x_{0})+\varepsilon \rbrack$. Therefore, the probabilities are ordered. The proof easily concludes by considering the supremum. The proof for a non-increasing contraction function follows a similar 
	argument. 
\end{proof}
\begin{remark}
	The inverse of a strictly monotone contraction function is an expansion function, i.e.,  a function where two images are more widely separated than the corresponding values of the images. Therefore if $\phi(X)=_{st} Y$ for a strictly monotone expansion function then $X\leq_{\text{wd}}Y$ holds.
\end{remark}
\begin{corollary} \label{cordisp}
	If $X \leq_{\text{disp}}  Y$ in the classical Lewis and Thompson sense, then $X\leq_{\text{wd}}Y$ holds.
\end{corollary}
\begin{proof}
	The proof follows directly from using jointly Theorem 1.7.4 of Müller and Stoyan (2002) and Theorem \ref{contrac}.    
\end{proof}

\begin{remark}
	Note, however, that the classical dispersive order is incompatible with almost all discrete distributions, including lattice and most empirical distributions \citep{ek-2024}.
	The following simple example shows that the weak dispersive order is weaker than the classical dispersive order: Let $U\{1,n\}$ denote the uniform distribution on $\{1,\ldots,n\}$. 
	Obviously, $U\{1,m\} \leq_{\text{wd}} U\{1,n\}$ for $m\leq n$, but they are not generally ordered with respect to the dispersive order in the L-T sense. For example, $U\{1,2\}$ and $U\{1,5\}$ are not ordered with respect to $\leq_{\text{disp}}$.
\end{remark}

As another example, we compare Bernoulli distributions.

\begin{example}
	Let $X \sim Be(p_1)$ and $Y \sim Be(p_2)$ be two Bernoulli random variables. A straightforward calculation shows that $X \leq_{\text{wd}} Y$ if and only if $\max\{p_1, 1-p_1\} \geq \max\{p_2, 1-p_2\}$. Then, $ Be(p) \leq_{\text{wd}} Be(0.5)$, for all $p\in [0,1]$.  
	
	On the other hand, it follows directly from Prop. 2.5 in \cite{ek-2024} that neither $X \leq_{\text{disp}} Y$ nor $Y \leq_{\text{disp}} X$ holds for $p_1\neq p_2$.
\end{example}

The next result directly follows from Theorem \ref{contrac}.

\begin{corollary} \label{cor:affin}
	Given a discrete random variable $X$, we obtain that
	\begin{enumerate}
		\item[a)] $X \leq_{\text{wd}}
		(\geq_{wd})\,  aX+b$, $\forall |a| >(<)\, 1$, $b\in \mathbb{R}$. 
		\item[b)] $X =_{wd} -X + b =_{wd} X + b$, $\forall b \in \mathbb{R}$.
	\end{enumerate}
\end{corollary}

\begin{remark} \label{rem:equiv} 
	Corollary \ref{cor:affin} points to an important difference between the weak and the usual dispersive order. The equivalence class of a distribution with cdf $F$ with respect to $=_{disp}$ is given by all shifts of $F$, i.e.\ $\{F(\cdot - b): b \in \R\}$ \citep[see][p.\ 157]{oja}, so distributions that are equivalent with respect to dispersion can only differ in location. Consequently, $-X=_{disp} X$ does not hold in general. The same holds for its discrete counterpart considered in Section \ref{sec:disc-disp}.
	
	Corollary \ref{cor:affin} b) shows that, besides shifts, also reflections belong to the equivalence class of a distribution with respect to $=_{wd}$. But this is not the complete equivalence class, as the following example shows. Consider random variables $X$ and $Y$ with support $\{1,\ldots,5\}$ and probability mass functions $p=(0.1,0.4,0.05,0.3,0.15)$ and $q=(0.4,0.1,0.25,0.15,0.1)$, respectively. Then, $X=_{wd} Y,$ but $X$ and $Y$ are not location shifts of one another.
\end{remark}

The following proposition and its proof closely follow \cite{levy37} and hold for all types of random variables—discrete, continuous, or otherwise.

\begin{proposition}\label{prop1b}
	Let $X$ and $Y$ be two independent random variables. 
	Then,
	\begin{eqnarray*}
		Q_{X+Y}(\varepsilon) \leq Q_{X}(\varepsilon), \, \forall \varepsilon > 0.
	\end{eqnarray*}
\end{proposition}
So $X+Y\geq_{wd}X$ holds.

\begin{proof}
	By considering the convolution, we obtain that
	\begin{eqnarray*}
		\mbox{Pr} \{ X+Y \in 
		\lbrack x_0, x_0+\varepsilon \rbrack \} & = &  \int_\mathbb{R} 
		\mbox{Pr} \{ X \in \lbrack x_0-y, x_0 - y+\varepsilon \rbrack
		\} dF_{Y}(y) \\
		& \leq & \int_\mathbb{R} 
		Q_{X}(\varepsilon) dF_{Y}(y) \\
		& = & Q_{X}(\varepsilon), \, \forall \varepsilon >0.
	\end{eqnarray*} 
	where the first inequality follows from the fact that the interval has length $\varepsilon$. Taking the supremum, the proof follows easily.
\end{proof}

\begin{example} \label{ex:convolution}
	By Prop. \ref{prop1b}, many well-known distributions can be ordered with respect to $\geq_{wd}$.
	\begin{enumerate}
		\item[a)]
		Let $X$ and $Y$ have Poisson distributions with parameters $\lambda$ and $\mu$. Then, $X \leq_{\text{wd}} Y$ if and only if $\lambda\leq\mu$.
		\item[b)]
		Let $X$ and $Y$ have binomial distributions $Bin(m,p)$ and $Bin(n,p)$. Then, $X \leq_{\text{wd}} Y$ if and only if $m\leq n$.
		\item[c)]
		Let $X$ and $Y$ have negative binomial distributions $NB(r,p)$ and $NB(s,p)$. Then, $X \leq_{\text{wd}} Y$ if and only if $r\leq s$.
		\item[d)]
		As an example of distributions with an arbitrary number of modes, consider the Hermite distributions $Herm(a,b), a,b>0$: if $U$ and $V$ are independent Poisson random variables with parameters $a$ and $b$, then $Z=U+2V \sim  Herm(a,b)$. Now, if $X \sim Herm(a_1,b_1), Y\sim Herm(a_2,b_2)$, and $X$ and $Y$ are independent, then $X+Y\sim Herm(a_1+a_2,b_1+b_2)$, where $a_1,a_2,b_1,b_2>0$.
		It follows that $Herm(a,b) \leq_{\text{wd}} Herm(c,d)$ if $a\leq c$ and $b\leq d$. It is worth recalling that the Hermite discrete random variable can model overdispersion in count data, (see, e.g., \cite{johnson:2005, kemp:1965}).
		
		As two examples, Figure \ref{fig:Hermite} shows the probability mass functions of three Hermite distributed random variables $X,$ $Y$ and $Z$ with $X <_{wd} Y <_{wd} Z$. 
		In the left panel, $a_1=a_2=0.10, a_3=0.15, b_1=0.10, b_2=b_3=0.15$; 
		in the right panel, $a_1=a_2=0.10, a_3=0.15, b_1=1.0, b_2=b_3=1.1$.
		
		\begin{figure}[ht]
			\centering
			\includegraphics[scale=0.5]{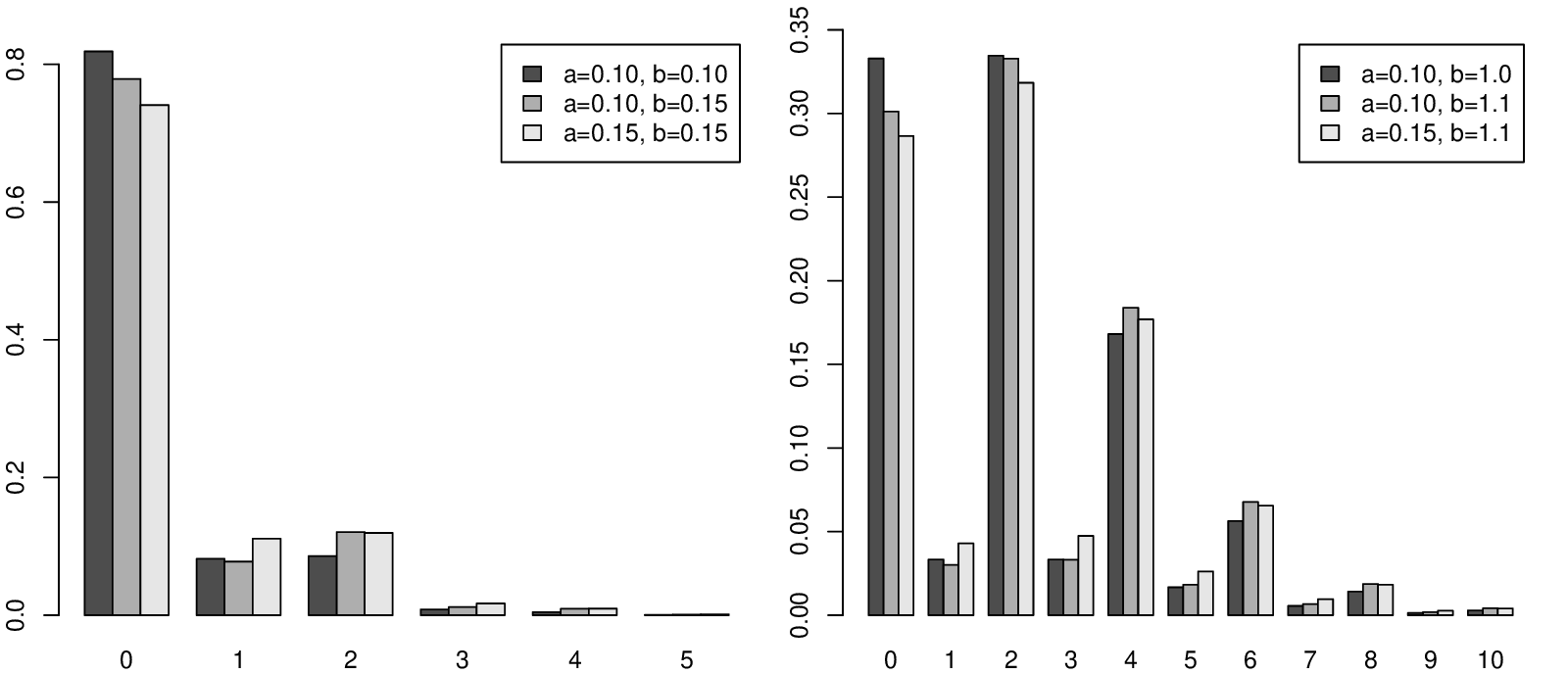}
			\caption{Both panels show probability mass functions of three Hermite distributed random variables $X$ (black), $Y$ (darkgrey) and $Z$ (lightgrey) with $X <_{wd} Y <_{wd} Z$} \label{fig:Hermite}
		\end{figure}
	\end{enumerate}
\end{example}

When the probability mass functions of two discrete distributions are decreasing, the following result can help to establish the weak dispersive order between the two distributions.

\begin{proposition} \label{prop-decreasing}
	Let $X$ and $Y$ be two random variables on $\mathbb{N}$ with decreasing probability mass functions. Further, assume $X\le_{st} Y$, where $\le_{st}$ denotes the usual stochastic order. Then, $X \leq_{\text{wd}} Y$ holds.
\end{proposition}

\begin{proof}
	We just need to take into account that the supremum is always achieved at $k=1$, i.e.,
	$$
	Q_{X}(\varepsilon) = \sup_{x_0} \mbox{Pr}\{ X\in [x_0, x_0+\varepsilon]\} 
	= F_{X}(1+\varepsilon) = F_{X}(k),  
	$$
	where $k$ is the maximum integer equal to or less than $1+\varepsilon$. The proof then concludes with the assumption that the two distributions are ordered in the stochastic sense. 
\end{proof}

\begin{example} \label{ex:decreasing}
	Let us consider two applications of Proposition \ref{prop-decreasing}.
	\begin{enumerate}
		\item[a)]
		Let $X \sim Geom(p_1)$ and $Y \sim Geom(p_2)$ be two geometric random variables, i.e. the number of trials needed to get one success. By Proposition \ref{prop-decreasing}, $X \leq_{we}  Y$ if and only if $p_1 \geq p_2$.
		\item[b)]
		Here, we consider the logarithmic distribution, which has many applications. For example, it has long been used to model species abundances \citep{fisher1943} and intensities (see, e.g., \citet{williams1964}).
		Let $X$ have a logarithmic distribution with probability mass function 
		$$
		f(k,p)=\frac{-1}{\log(1-p)} \, \frac{p^k}{k}, \quad k\geq 1,
		$$ 
		and let $Y$ be logarithmically distributed with probability mass function $f(k,q), k\geq 1$, where $0<p,q<1$. 
		For the logarithmic distribution, it is easy to see that $X\leq_{lr} Y$, where $\leq_{lr}$ denotes the likelihood ratio order, if $p\leq q$. 
		Since the likelihood ratio order implies the stochastic order (see, e.g., Müller and Stoyan 2002, p. 12), we obtain $X\leq_{st} Y$ if $p\leq q$.
		On the other hand, if $X\leq_{st} Y$, then $EX=-[\log(1-p)]^{-1} p/(1-p)\leq EY$, which implies $p\leq q$. Thus, $X \leq_{st} Y$ if and only if $p\leq q$.
		
		Using Proposition \ref{prop-decreasing}, it immediately follows that $X \leq_{\text{wd}} Y$ if and only if $p\leq q$.
	\end{enumerate}
\end{example}

\section{Relation to the discrete dispersive order of Eberl and Klar} \label{sec:disc-disp}

We begin with some definitions and results necessary to understand the discrete dispersive order of \cite{ek-2024}, which serves as a fundamental order for defining measures of dispersion for discrete distributions in an axiomatic way.

\begin{definition}
	\label{def:PDdists}
	Let $\Ddists$ denote the set of discrete distributions, and let $F \in \Ddists$ be a cdf with probability mass function $f$ and let $X \sim F$. The class $\PDdists \subseteq \Ddists$ is defined by
	\begin{align*}
		F \in \PDdists \Leftrightarrow& \ \supp(F) \text{ is order-isomorphic to a subset of } \Z \text{ with at least two }\\
		& \, \text{ elements} \\
		\Leftrightarrow & \ \exists \text{ a subset } A \subseteq \Z, |A| \geq 2, \text{and a bijection } \varphi: \supp(F) \to A\\
		& \ \text{such that } x \leq y \Leftrightarrow \varphi(x) \leq \varphi(y) \ \forall x, y \in \supp(F).
	\end{align*}
\end{definition}

\begin{proposition}
	\label{thm:IS-char}
	Define $\indsetset = \{\Z, \N, -\N\} \cup \{ \{1, \ldots, n\}: n \in \Nz \}$ and
	\begin{align*}
		\idnseqset_A =& \Big\{ (x_j, p_j)_{j \in A} \subseteq \R \times (0, 1]: x_i < x_j \text{ for } i < j, p_j > 0 \text{ for } j \in A, \sum_{j \in A} p_j = 1 \Big\}\\
		\intertext{for $A \in \indsetset \setminus \Z$ as well as}
		\idnseqset_{\Z} =& \Big\{ (x_j, p_j)_{j \in \Z} \subseteq \R \times (0, 1]: x_i < x_j \text{ for } i < j, p_j > 0 \text{ for } j \in \Z, \sum_{j \in \Z} p_j = 1,\\
		& \; \hphantom{\Big\{ (x_j, p_j)_{j \in \Z} \subseteq \R \times (0, 1]:} \inf\{j \in \Z: \sum_{i \leq j} p_i \geq \tfrac{1}{2}\} = 0 \Big\}.
	\end{align*}
	For any $F \in \PDdists$, there exists a unique index set $A \in \indsetset$ that is order-isomorphic to $\supp(F)$, and there exists a unique sequence 
	$(x_j, p_j)_{j \in A} \in \idnseqset_A$ such that $\Pr(X = x_j) = p_j$ for all $j \in A$.
	This unique association is denoted by $F \cdfis \left( A, (x_j, p_j)_{j \in A} \right)$. $A$ is said to be the {\em indexing set} of $F$ and $(x_j, p_j)_{j \in A}$ is said to be the {\em identifying sequence} of $F$.
\end{proposition}

In the following, let $F \cdfis \left( A, (x_j, p_j)_{j \in A} \right)$ and $G \cdfis$ $\left( B, (y_j, q_j)_{j \in B} \right)$. Furthermore, we use the conventions $x_a = -\infty$ and $F(x_a) = 0$ for $a < \min A$ as well as $x_a = \infty$ and $F(x_a) = 1$ for $a > \max A$, provided that the minimum and the maximum exist, respectively. \cite{ek-2024} defined the following relations, where $\smi{A}=A \setminus \{\min A\}$.

\begin{definition} \label{def:compj}
	Let $F, G \in \PDdists$.
	\begin{enumerate}
		\item[a)] The relation $\stdCompj$ on the set $A \times B$ is defined by
		\begin{align*}
			a \stdCompj b &\Leftrightarrow (F(x_{a-1}), F(x_a)) \cap (G(y_{b-1}), G(y_b)) \neq \emptyset.
		\end{align*}
		for $a \in A, b \in B$. The set $R(\stdCompj)$ of all $(a, b) \in A \times B$ with $a \stdCompj b$ is said to be the set of $(F, G)$-{\em dispersion-relevant} pairs of indices.
		\item[b)] The relation $\stdCompja$ on the set $\smi{A} \times \smi{B}$ is defined by
		\begin{equation*}
			a \stdCompja b \Leftrightarrow (a \stdCompj b) \land (a-1 \stdCompj b-1)
		\end{equation*}
		for $a \in \smi{A}, b \in \smi{B}$. The set $R(\stdCompja)$ of all $(a, b) \in \smi{A} \times \smi{B}$ with $a \stdCompja b$ is said to be the set of $(F, G)$-{\em $\land$-dispersion-relevant} pairs of indices.
	\end{enumerate}
\end{definition}

\cite{ek-2024} proposed the following discrete dispersive order. 

\begin{definition} \label{def:disco}
	Let $F, G \in \PDdists$. $G$ is said to be {\em at least as $\land$-discretely dispersed} as $F$, denoted by $F \discoa G$, if the following two conditions are satisfied:
	\begin{enumerate}
		\item[(i)] $q_b \leq p_a \quad \forall (a, b) \in R(\stdCompj)$,
		\item[(ii)] $x_a - x_{a-1} \leq y_b - y_{b-1} \quad \forall (a, b) \in R(\stdCompja)$.
	\end{enumerate}
\end{definition}

Now, we can formulate the main result of this section. It shows that for distributions in $\PDdists$, the weak dispersive order is indeed weaker than the discrete dispersive order.

\begin{theorem} \label{th:disc-disp}
	Let $F, G \in \mathcal{D}_0$, where $\mathcal{D}_0$ is defined as in \citet[][Def. 3.1]{ek-2024}. Then, $F \discoa G$ implies $F \leq_{\text{wd}} G$.
	
	\begin{proof}
		In this proof, we utilize the notation from \citet{ek-2024}. In particular, let $F \cdfis \left( A, (x_j, p_j)_{j \in A} \right)$ and $G \cdfis \left( B, (y_j, q_j)_{j \in B} \right)$.
		
		In the following, we will prove that, for all $y^* \in \R$ and all $\varepsilon > 0$, there exists a $x^* \in \R$ such that $F(x^* + \varepsilon) - F(x^*) \geq G(y^* + \varepsilon) - G(y^*)$. This then directly implies $F \leq_{\text{wd}} G$.
		
		Let $y^* \in \R$ and $\varepsilon > 0$. Choose $b, \beta \in B \cup \{\min B - 1\}$ with $b \leq \beta$ such that $y_{b+1} = \sup\{y \in \R: G(y) = G(y^*)\}$ and $y_\beta = \inf\{y \in \R: G(y) = G(y^* + \varepsilon)\}$. Since the case $b = \beta$ is trivial, we assume $b < \beta$. Furthermore, choose $a \in A \cup \{\min A - 1\}$ such that $F(x_a) = \max\left( [0, G(y_b)] \cap F(\R) \right)$. Then, we have $F(x_a) \leq G(y_b) < F(x_{a+1})$.
		
		Let $k_{a+1} = b+1$ and note that $a+1 \stdCompja k_{a+1}$. Now, exactly one of the following two conditions is satisfied:
		\begin{enumerate}
			\item[(i)] $G(y_{b+1}) \leq F(x_{a+1})$,
			\item[(ii)] $G(y_{b+1}) > F(x_{a+1})$. 
		\end{enumerate}
		In case (i) under the condition $F(x_{a+1}) < 1$, we choose $k_{a+2} \in B \cap [k_{a+1}+1, \infty)$ such that $G(y_{k_{a+2}-1}) = \min \nn{a+2}{F}{G}$. This implies $G(y_{k_{a+2}-1}) \leq F(x_{a+1})$ and $a+2 \stdCompja k_{a+2}$. If $F(x_{a+2}) < 1$, we choose $k_{a+3} \in B$ such that $G(y_{k_{a+3}-1}) = \min \nn{a+3}{F}{G}$. According to \citet[][Lemma A.2 in the supplement]{ek-2024}, we have $k_{a+3} > k_{a+2}$. As long as $F(x_{a+i}) < 1, i = 2, 3, \ldots,$ we can keep on choosing $k_{a+i+1} \in B \cap [k_{a+i}+1, \infty)$ such that $G(y_{k_{a+i+1}-1}) = \min \nn{a+i+1}{F}{G}$.
		
		In case (ii), we choose $k_{a+2} = b+2$, which is equivalent to $G(y_{k_{a+2}-1}) = \max \nn{a+2}{F}{G}$. (Note that $G(y_{b+1}) < 1$, because assuming $G(y_{b+1}) = 1$ implies $F(x_{a+2}) = G(y_{b+1}) = 1$ and thus $p_{a+2} < q_{b+1}$ in spite of $a+2 \stdCompj b+1$, which contradicts $F \discoa G$.) This implies $G(y_{k_{a+2}-1}) > F(x_{a+1})$ and $a+2 \stdCompja k_{a+2}$ as well as
		\begin{align*}
			0 &< G(y_{k_{a+2}-1}) - F(x_{a+1})\\
			&= G(y_{b+1}) - F(x_{a+1})\\
			&= (G(y_b) - F(x_a)) + (q_{b+1} - p_{a+1})\\
			&\leq G(y_b) - F(x_a).
		\end{align*}
		Then, we make the same distinction as previously, just increasing the indices by one. Again, exactly one of the following conditions is true:
		\begin{enumerate}
			\item[(i')] $G(y_{b+2}) \leq F(x_{a+2})$,
			\item[(ii')] $G(y_{b+2}) > F(x_{a+2})$. 
		\end{enumerate}
		In case (i') under the condition $F(x_{a+2}) < 1$, we proceed as in case (i), just starting with one index higher, i.e. by choosing $k_{a+3} \in B \cap [k_{a+2}+1, \infty)$ such that $G(y_{k_{a+3}-1}) = \min \nn{a+3}{F}{G}$. In case (ii'), we proceed as in case (ii), just starting with one index higher, i.e. by choosing $k_{a+3} = b+3$.
		
		Overall, we iteratively obtain
		\begin{equation}
			\label{eqn:iterResult1}
			\forall j \in A \cap [a+1, \infty) \; \exists \, \text{pairwise distinct } k_j \in B \cap [b+1, \infty): j \stdCompja k_j. 
		\end{equation}
		In this iteration, we either remain in cases (ii), (ii'), (ii''), \ldots\ or we at some point get into cases (i), (i'), (i''), \ldots\ and remain there from that point on out. Let $\tilde{\beta} = \max\{k \in B \cap (-\infty, \beta]\,|\, \exists j \in A: k = k_j\}$ and define $\alpha \in A$ by $\tilde{\beta} = k_{\alpha}$. Then, (\ref{eqn:iterResult1}) can be made more specific in the following sense:
		\begin{equation*}
			\label{eqn:iterResult2}
			\forall j \in \{a+1, \ldots, \alpha\} \; \exists \, \text{pairwise distinct } k_j \in \{b+1, \ldots, \tilde{\beta}\}: j \stdCompja k_j. 
		\end{equation*}
		It follows that
		\begin{align*}
			x_{\alpha} - x_{a+1} &= \sum_{j=a+2}^{\alpha} x_j - x_{j-1}\\
			&\leq \sum_{j=a+2}^{\alpha} y_{k_j} - y_{k_j-1} \leq \sum_{k=b+2}^{\tilde{\beta}} y_k - y_{k-1}\\
			&= y_{\tilde{\beta}} - y_{b+1} \leq y_\beta - y_{b+1} < \varepsilon.
		\end{align*}
		Choose $\delta > 0$ small enough such that $x^* = x_{a+1} - \delta$ and $x^* + \varepsilon \geq x_\alpha$.
		
		If the pair of $\alpha$ and $k_\alpha = \tilde{\beta}$ is in case (i'\ldots'), then we have $G(y_{\tilde{\beta}-1}) \leq F(x_{\alpha-1}) < G(y_{\tilde{\beta}}) \leq G(y_\beta)$. If $F(x_\alpha) = 1$, then $G(y_\beta) \leq F(x_\alpha)$ obviously holds; if $F(x_\alpha) < 1$, then $G(y_{\beta-1}) < \min \nn{\alpha+1}{F}{G}$ and thus $G(y_\beta) = \min \nn{\alpha+1}{F}{G} \leq F(x_\alpha)$ holds because of the maximality of $\tilde{\beta}$. Overall, it follows that
		\begin{align*}
			F(x^* + \varepsilon) - F(x^*) &\geq F(x_\alpha) - F(x_a) \geq G(y_\beta) - G(y_b) = G(y^* + \varepsilon) - G(y^*),
		\end{align*}
		concluding the proof if the pair of $\alpha$ and $k_\alpha = \tilde{\beta}$ is in case (i'\ldots').
		
		If the pair of $\alpha$ and $k_\alpha = \tilde{\beta}$ is in case (ii'\ldots'), then we have $F(x_{\alpha-1}) < G(y_{\tilde{\beta}-1}) < F(x_\alpha)$ and $G(y_{\tilde{\beta}-1}) - F(x_{\alpha-1}) \leq G(y_b) - F(x_a)$. If $F(x_\alpha) = 1$, then $G(y_\beta) \leq F(x_\alpha)$ obviously holds. If $F(x_\alpha) < 1$, then either $\alpha+1$ and $k_{\alpha+1} > \beta$ are in case (i'\ldots'), which implies $G(y_\beta) \leq G(y_{k_{\alpha+1}-1}) \leq F(x_\alpha)$; or $\alpha+1$ and $k_{\alpha+1} = \beta+1$ are in case (ii'\ldots'), which implies $0 < G(y_{\beta}) - F(x_{\alpha}) \leq G(y_b) - F(x_a)$. Overall, it follows that
		\begin{align*}
			F(x^* + \varepsilon) - F(x^*) &\geq F(x_\alpha) - F(x_a)\\
			&= (G(y_\beta) - G(y_b)) + (G(y_b) - F(x_a)) - (G(y_\beta) - F(x_\alpha))\\
			&\geq G(y_\beta) - G(y_b) = G(y^* + \varepsilon) - G(y^*),
		\end{align*}
		concluding the proof if the pair of $\alpha$ and $k_\alpha = \tilde{\beta}$ is in case (ii'\ldots').
	\end{proof}
\end{theorem}

The next example shows that several common dispersion measures do not preserve the weak dispersive order.

\begin{example}
	The interquartile range (IQR) does not preserve the order $\discoa$ \citep[Theorem 4.7]{ek-2024}. Since, by Theorem~\ref{th:disc-disp}, the order $\leq_{\text{wd}}$ is weaker than $\discoa$, it follows that the IQR also does not preserve $\leq_{\text{wd}}$.
	
	Other common measures of dispersion do preserve $\discoa$ \citep{ek-2024}. However, the following three-point distributions illustrate that this is not generally true for $\leq_{\text{wd}}$. Consider two distributions:  
	\[
	p = \{p_0 = 0.6,\ p_1 = 0.2,\ p_2 = 0.2\}, \quad
	q = \{q_0 = 0.3,\ q_1 = 0.5,\ q_2 = 0.2\}.
	\]
	It is straightforward to verify that $p \leq_{\text{wd}} q$.
	On the other hand, the dispersion measures for these distributions yield:
	\begin{itemize}
		\item Standard deviations: $\mathrm{SD}_p = 0.80$, $\mathrm{SD}_q = 0.65$;
		\item Mean absolute deviations from the mean: $\mathrm{MAD}_p = 0.72$, $\mathrm{MAD}_q = 0.54$;
		\item Gini mean differences: $\mathrm{GMD}_p = 0.80$, $\mathrm{GMD}_q = 0.74$.
	\end{itemize}
	All differences $\mathrm{SD}_p - \mathrm{SD}_q$, $\mathrm{MAD}_p - \mathrm{MAD}_q$, and $\mathrm{GMD}_p - \mathrm{GMD}_q$ are strictly positive. Based on these classical measures, the distribution $p$ would be considered to be more dispersed than $q$, despite $p \leq_{\text{wd}} q$.
\end{example}

As the empirical illustrations in Section~\ref{sec6} show, such discrepancies are relatively rare but stem from deeper conceptual differences, which are discussed in further detail in the next section.

\section{New measures of variability based on the concentration function} \label{sec5}

The preceding example highlights an important conceptual distinction: the Lévy concentration function captures concentration rather than deviation from a central location. Traditional dispersion measures such as the standard deviation (SD) and the median absolute deviation (MAD) quantify variability with respect to a fixed location parameter, typically the mean or median. In contrast, the Gini mean difference (GMD) assesses variability by averaging the absolute differences between all pairs of independent copies of the random variable. Formally, for a random variable \( X \), the GMD is defined as:
\[
\mathrm{GMD}_X = \mathbb{E}_{X'}\left[ \mathbb{E}_{X|X'} \left[ |X - X'| \right] \right],
\]
where \( X' \) is an independent copy of \( X \).

In the example under discussion, although all three classical dispersion measures are larger for distribution \( p \) than for distribution \( q \), both distributions share the same support. Thus, the higher concentration of \( p \) indicates a lower level of uncertainty or randomness. From a practical perspective, if one had to predict a likely value, choosing for \( p \) would be more straightforward than for \( q \), suggesting that \( p \) is more ``predictable''. This interpretation is supported by the Shannon entropy values: \( H(p) = -\sum p_i \log_2(p_i) = 1.371 \), while \( H(q) = 1.485 \). The lower entropy for \( p \) reflects reduced uncertainty.

This relationship is underpinned by the fact that the vector $p$ majorizes $q$ (i.e. $q \preceq p$); that is, 
$\sum_{i=1}^k x_{[i]} \geq \sum_{i=1}^k y_{[i]}$ for all $k$, where $x_{[i]}$ denotes the $i$-th largest entry of $x$. 
Since entropy is Schur-concave - that is, $f(q) \geq f(p)$ if $q \preceq p$ - it decreases as the distribution becomes more concentrated. 
Hence, in this example, greater concentration corresponds to lower entropy, providing a compelling case where classical dispersion measures may fail to capture the perceived variability in discrete distributions. This underscores the relevance of concentration-based and entropy-based approaches, particularly when comparing distributions with identical support.

However, concentration can be viewed both as a lower bound for several classical measures of variability - those that quantify deviation from a central point - and as a dispersion measure in its own right. Using the formula  
\[ \mathbb{E}[Z^r] = \int_0^\infty rt^{r-1} \Pr(Z > t)\,dt \]  
for a non-negative random variable $Z$ and $r \geq 1$, we obtain
$$
E[|X-a|^r] \geq r\int_{0}^{\infty}t^{r-1}(1-Q_X(2t))dt, \quad \forall a \in \mathbb{R}, \, \forall r\geq 1.
$$
Therefore, 
$$
E[|X-a_r|^r] \geq \frac{r}{2^r} \int_{0}^{\infty} \varepsilon^{r-1} (1-Q_X(\varepsilon)) \, d\varepsilon, \quad  \forall r \geq 1,
$$
where $a_r = \arg \min\limits_{a} E[|X-a|^r]$, $r\geq 1$. In fact, any functional $\nu_r(X)$ defined as 
$$
\nu_r(X) = \frac{\sqrt[r]{r}}{2} \left ( \int_{0}^{\infty} \varepsilon^{r-1} (1-Q_X(\varepsilon)) d\varepsilon \right )^{1/r}, 
$$
can also be used as a measure of the variability of $X$, preserving the weak dispersive order $\leq_{\text{wd}}$.

To describe the properties of this functional, recall that a measure of variability $\nu$ is a map from the set of random variables to $\mathbb{R}$, such that given a random variable $X,$ $\nu(X)$  quantifies the variability of $X$. Next, we list a number of properties that a measure of variability should reasonably satisfy: 

\begin{itemize}
	\item{P1}  Law invariance: if $X$ and $Y$ have the same distribution, then $\nu(X)=\nu(Y).$ 
	\item{P2} Translation invariance: $\nu(X+k)=\nu(X)$ for all $X$ and all constant $k$.
	\item{P3} Absolute homogeneity: $\nu(\lambda X)= |\lambda| \nu(X)$ for all $X$ and all $\lambda\in\R$.
	\item{P4} Non-negativity: $\nu(X) \geq 0$, with  equality if $X$ is degenerate at $c \in \mathbb{R}$.
	\item{P5} Consistency with some dispersion order $\leq_D$: if $X \leq_{D}Y$, then $\nu(X) \leq \nu(Y)$. 
\end{itemize}

A functional $\nu$ satisfying the properties P1-P5  - using $\leq_{\text{disp}}$ as the dispersion order in P5 - is called a measure of variability or spread in the sense of Bickel and Lehmann; see \cite{BiLe-1979}. It is obvious that $\nu_r(X)$ satisfies P1, P2, and P4. To prove P3, we need only observe that $Q_{\lambda X}(\varepsilon)=Q_{X}(\varepsilon/ |\lambda|)$, from which $\nu_r(\lambda X)=|\lambda |\nu_r(X)$, $\lambda \in \mathbb{R}$, follows. Property P5 - whether using $\leq_{\text{disp}}$ or the discrete dispersive order $\discoa$ - follows directly from Corollary \ref{cordisp} and Theorem \ref{th:disc-disp}, respectively.

Note that for many discrete distributions, $\nu_r(X)$ is not difficult to compute. For a random variable $X$ with values in $\mathbb{N}_0$, we obtain 
\begin{align*}
	\nu_r(X) &= \frac{1}{2} \left ( \sum_{k=0}^{\infty} \left( (k+1)^r - k^r \right) (1-Q_X(k)) \right )^{1/r}.
\end{align*}
and, in particular, $\nu_1(X) = 1/2 \sum_{k=0}^{\infty} (1-Q_X(k))$.
Thus, $\nu_r$ is essentially the sum over the minimum probability outside an interval of length $k$; if this is small, the distribution is highly concentrated and its variability is low. 

\begin{example}
	\begin{enumerate}
		\item[a)]
		If $X\sim Be(p)$, we obtain $\nu_1(X) = \min\{p, 1-p\}/2, \ \nu_2(X) = \sqrt{\min\{p, 1-p\}}/2$. 
		\item[b)]
		Let $X \sim Geom(p)$ with $F_X(k)=1-(1-p)^k, k\in\mathbb{N}$. Then, $1-Q_X(k)=(1-p)^{k+1}$, which entails 
		$\nu_1(X)=(1-p)/(2p),  \ \nu_2(X)=\sqrt{(1-p)(2-p)}/(2p).$
	\end{enumerate}
\end{example}

\begin{definition}
	\begin{enumerate}
		\item[a)] Let $X$ and $Y$ be discrete random variables and let $p$ and $q$ be their probability mass functions (defined on the union of their supports). $X$ is said to be {\itshape less random} than $Y$ (written $X\leq_{rand} Y$) if $q \preceq p$, where $\preceq$ denotes the order of majorization \citep{hickey:1983}. 
		\item[b)] A distribution $p$ with support on the lattice of integers is unimodal if there exists at least one integer $M$ such that 
		$$
		p_k\geq p_{k-1}, \ \forall k\leq M, \qquad   p_{k+1}\leq p_k, \ \forall k\geq M.
		$$
	\end{enumerate}
\end{definition}

Closely related characterizations of discrete unimodality are studied in \cite{berthe1984}. To compare distributions in terms of randomness, \cite{hickey:1984} and \cite{hickey:1986} introduce a notion of majorization for absolutely continuous distributions via the decreasing rearrangement of their density functions. For the absolutely continuous case, \cite{fersur2003} shows that comparison in terms of randomness is equivalent to the weak dispersive order for unimodal distributions, as defined in \cite{sudku1988}. The following proposition serves an analogous purpose for discrete distributions.

\begin{proposition} \label{prob:rand}
	Let $p$ and $q$ be unimodal distributions on the non-negative integers, and let $X$ and $Y$ be distributed according to $p$ and $q$. Then,
	\begin{align*}
		X\leq_{\text{wd}} Y &\Leftrightarrow X\leq_{rand} Y.
	\end{align*}
\end{proposition}

\begin{proof}
	By the unimodality assumption, 
	\begin{align*}
		\max_{i\in\mathbb{N}_0} \left\{ \sum_{k=i}^{i+m} p_k \right\} &= p_{[0]}+\ldots+p_{[m]},
	\end{align*}
	where $p_{[0]}\geq p_{[1]} \geq \ldots$ are the components of $p$ in decreasing order, and the assertion follows.
\end{proof}

\begin{remark} \label{rem:unimod}
	\begin{enumerate}
		\item[a)] With the help of Prop. \ref{prob:rand}, one can recover the results in Example \ref{ex:convolution} a)-c) and Example \ref{ex:decreasing} a), since it is known that the corresponding distributions are ordered with respect to randomness \citep{hickey:1983}.
		\item[b)] 
		Consider two unimodal distributions on the non-negative integers with $p_{[i]}=q_{[i]}, i\in\mathbb{N}_0$.
		Then, $X=_{rand} Y$, and $X =_{wd} Y$ follows by Prop. \ref{prob:rand}. Hence, all unimodal distributions with the same ordered probabilities are equivalent with respect to $=_{wd}$.
		\item[c)] 
		Of course, without the unimodality assumption, the orders $\leq_{rand}$ and $\leq_{\text{wd}}$ are not equivalent. Consider the distributions $X \equiv \{p_0=0.1, p_1=0.4, p_2=0.4, p_3=0.1\}$ and $Y \equiv \{q_0=0.4, q_1=0.1, q_2=0.1, q_3=0.4\}$. It is clear that $X=_{rand} Y$, but $X \leq_{\text{wd}} Y$ while $Y \not \leq_{\text{wd}} X$.
		\item[d)]
		\begin{sloppypar}
			Consider the class of unimodal distributions whose support is a subset of $\{1,\ldots,M\}$. 
			In this class, the discrete uniform distribution $U\{1,M\}$ is the unique maximizer with respect to $\nu_1$, with 
		\end{sloppypar}
		\begin{align*}
			\nu_1 &= \frac{1}{2} \sum_{k=0}^{M-2} (1-Q_X(k)) 
			= \frac{1}{2} \sum_{k=1}^{M-1} \left(1-\frac{k}{M} \right) = \frac{M-1}{ 4}.
		\end{align*}
		Without the unimodality assumption, there is no maximal element: the two-point distribution with mass $1/2$ at $1$ and $M$ and the uniform distribution are not comparable with respect to $\leq_{\text{wd}}$. Note that for the latter distribution $\nu_1=(M-1)/4$ holds, as it does for the uniform distribution.
		\item[e)]
		The measure $\nu_r$ is finite if $\mathbb{E}|X|^r < \infty$, i.e., it is not robust with respect to outliers. 
		However, it is straightforward to define robust dispersion measures based on a similar construction. For a random variable $X$ taking values in $\mathbb{N}_0$, consider the measure  
		\[
		\nu_{\text{rob}}(X) = \sum_{k=0}^{\infty} \frac{1 - Q_X(k)}{1+k^2},
		\]
		which satisfies properties P1, P2, P4, and property P5 with respect to $\leq_{\text{disp}}$, $\discoa$ and $\leq_{\text{wd}}$ (note that P3 is not meaningful for lattice distributions). 
		
		Let $x_1, \ldots, x_n$ be an independent sample from $X$, and define $M = \max\{x_1, \ldots, x_n\}$. Then, $1 - Q(k) = 0$ for all $k \geq M$. Suppose now that a new observation $x>M$ is added. The empirical influence function satisfies  
		\[
		n\left(\nu_{\text{rob}}(x_1, \ldots, x_{n-1}, x) - \nu_{\text{rob}}(x_1, \ldots, x_n)\right) 
		= c + n \sum_{k=M}^x \frac{1/n}{1+k^2} \leq c + \sum_{k=0}^{\infty} \frac{1}{1+k^2},
		\]
		for some constant $c$. That is, the empirical influence function is bounded, demonstrating the robustness of $\nu_{\text{rob}}$ with respect to outliers.
	\end{enumerate}
\end{remark}

\section{Empirical illustrations} \label{sec6}

\subsection{Counts of swimbladder nematodes in Japanese eels} \label{ex1}
The first example compares counts of larvae of swimbladder nematodes in two populations of the Japanese eel ({\itshape Anguilla japonica}) from southwest Taiwan. Münderle et al. (2006) compared wild eels from the Kao-Ping River (sample 2, $n=168$) with cultured eels from an adjacent aquaculture farm (sample 1, $n=71$). All recorded nematodes belong to the species {\itshape Anguillicoloides crassus}.
Figure \ref{fig:ex1} shows bar plots of the relative frequencies $p$ and $q$ for the two datasets. Table \ref{tab:ex1} presents the absolute frequencies $h_i (i=1,2)$ alongside the relative frequencies for both datasets.
We see that $p_0<q_0$,but $p_1>q_1$, $p_2>q_2$, etc., and it is not possible to deduce anything regarding the weak dispersive ordering between the two samples based directly on these values.
\begin{table}[ht]
	\setlength{\tabcolsep}{2pt}
	\centering
	\begin{tabular}{rrrrrrrrrrrrrrr}  \hline
		$k$ & 0 & 1 & 2 & 3 & 4 & 5 & 6 & 7 & 8 & 12 & 16 & 21 & 42 & 64 \\  \hline
		$h_1$ & 32 & 15 & 8 & 4 & 1 & 1 & 3 & 2 & 1 & 0 & 1 & 1 & 1 & 1 \\ 
		$p$ & 0.45 & 0.21 & 0.11 & 0.06 & 0.01 & 0.01 & 0.04 & 0.03 & 0.01 & 0.00 & 0.01 & 0.01 & 0.01 & 0.01 \\ \hline
		$h_2$ & 134 & 19 & 9 & 0 & 1 & 2 & 0 & 1 & 1 & 1 & 0 & 0 & 0 & 0 \\ 
		$q$ & 0.80 & 0.11 & 0.05 & 0.00 & 0.01 & 0.01 & 0.00 & 0.01 & 0.01 & 0.01 & 0.00 & 0.00 & 0.00 & 0.00 \\ \end{tabular}
	\caption{Absolute frequencies $h_1$ and $h_2$ and relative frequencies $p$ and $q$ for the datasets in Example \ref{ex1} \label{tab:ex1} }
\end{table}

\begin{figure}[ht]
	\begin{center}
		\includegraphics[width=\textwidth]{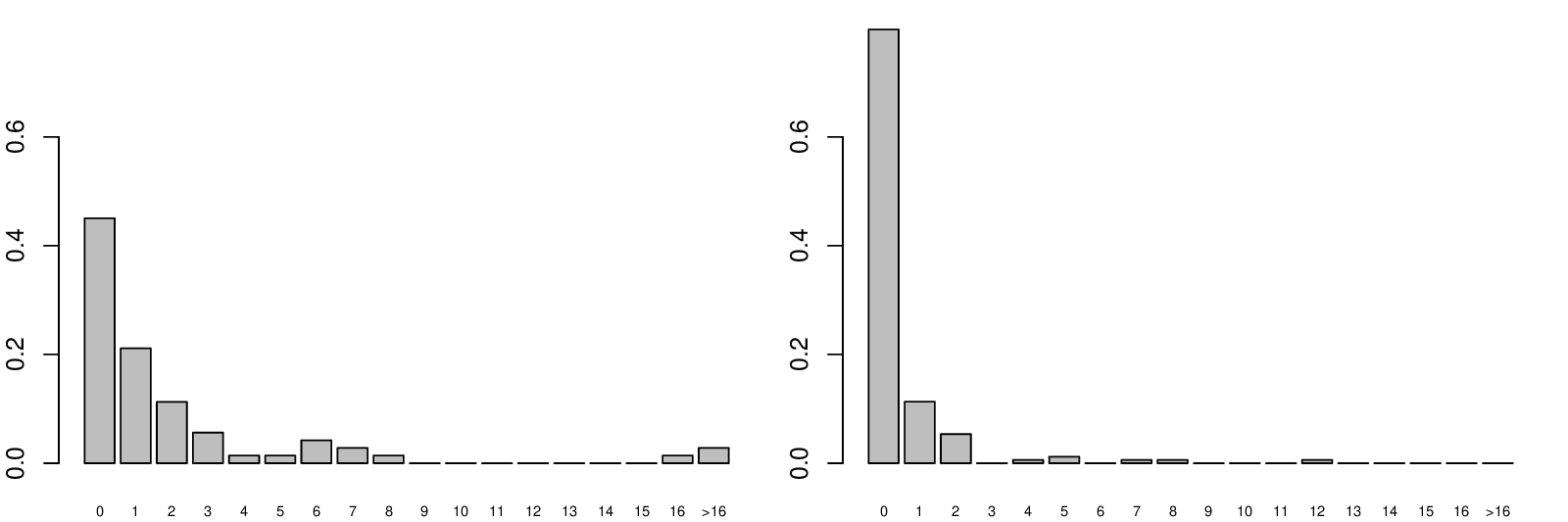}
	\end{center}
	\caption{Bar plots of the relative frequencies for the datasets \\ in Example \ref{ex1} \label{fig:ex1} }
\end{figure}

\smallskip
{\bfseries Findings:} For these two samples, both $X >_{st} Y$ and $X >_{wd} Y$ hold. Consequently, for any location measure $\mu$, we have $\mu(X) \geq \mu(Y)$. Similarly, for the dispersion measures $\nu_r$ and $\widehat{\nu}_{\text{rob}}$ introduced in Section \ref{sec5}, the inequality $\nu(X) \geq \nu(Y)$ also holds.  
Since the two samples are not comparable with respect to $\discoa$, no definitive conclusion can be drawn regarding the ordering of SD, MAD, and GMD. Nevertheless, Table \ref{tab:cex1} shows that even for these classical measures, as well as for the interquartile range, the inequality $\nu(X) \geq \nu(Y) $ still holds:

\begin{table}[ht]
	{\smallskip\centering
		\begin{tabular}{c|ccccccc}
			& SD & MAD & GMD & IQR & $\widehat{\nu}_1$ & $\widehat{\nu}_2$ & $\widehat{\nu}_{\text{rob}}$ \\ \hline
			sample 1 & 9.39 & 4.29 & 5.48 & 2 & 1.65 & 4.95 & 0.90 \\
			sample 2 & 1.43 & 0.74 & 0.83 & 0 & 0.23 & 0.75 & 0.51
		\end{tabular}
		\par}
	\caption{Dispersion measures for the datasets in Example \ref{ex1} \label{tab:cex1} }
\end{table}

\bigskip
The following two examples consider swimbladder and intestinal nematodes in European eels ({\itshape Anguilla anguilla}). The differences between the two distributions are rather small in Example \ref{ex2} and large in Example \ref{ex3}.

\subsection{Counts of swimbladder nematodes in European eels} \label{ex2}
As second example, we consider counts of larvae of swimbladder nematodes in European eels from two different locations: sample 1 (n = 196) from the River Rhine near Karlsruhe, and sample 2 ($n=100$) from the River Rhine near Sulzbach \citep{munderle2006}. Again, all recorded nematodes were A. crassus. \cite{klar:2010} concluded that both samples can be assumed to come from the same distribution. For example, a two-sample Kolmogorov-Smirnov test yields a $p$-value of 0.8.
Figure \ref{fig:ex2} displays bar plots of the relative frequencies $p$ and $q$ for the two datasets. Table \ref{tab:ex2} presents the absolute frequencies $h_i \, (i=1,2)$ alongside the relative frequencies for both datasets.
\begin{table}[ht]
	\setlength{\tabcolsep}{2pt}
	\centering
	\begin{tabular}{rrrrrrrrrrrrrrrrrr}  \hline
		$k$ & 0 & 1 & 2 & 3 & 4 & 5 & 6 & 7 & 8 & 9 & 10 & 11 & 12 & 13 & 16 & 22 & 23 \\ \hline
		$h_1$ & 104 & 47 & 16 & 13 & 5 & 3 & 2 & 1 & 0 & 1 & 0 & 0 & 0 & 1 & 1 & 1 & 1 \\ 
		$p$ & 0.53 & 0.24 & 0.08 & 0.07 & 0.03 & 0.01 & 0.01 & 0.00 & 0.00 & 0.00 & 0.00 & 0.00 & 0.00 & 0.00 & 0.00 & 0.00 & 0.00 \\ \hline
		$h_2$ & 61 & 16 & 10 & 1 & 2 & 1 & 1 & 1 & 2 & 0 & 2 & 2 & 1 & 0 & 0 & 0 & 0 \\ 
		$q$ & 0.61 & 0.16 & 0.10 & 0.01 & 0.02 & 0.01 & 0.01 & 0.01 & 0.02 & 0.00 & 0.02 & 0.02 & 0.01 & 0.00 & 0.00 & 0.00 & 0.00 \\ 
	\end{tabular}
	\caption{Absolute frequencies $h_1$ and $h_2$ and relative frequencies $p$ and $q$ for the datasets in Example \ref{ex2} \label{tab:ex2}}
\end{table}

\begin{figure}[ht]
	\begin{center}
		\includegraphics[width=\textwidth]{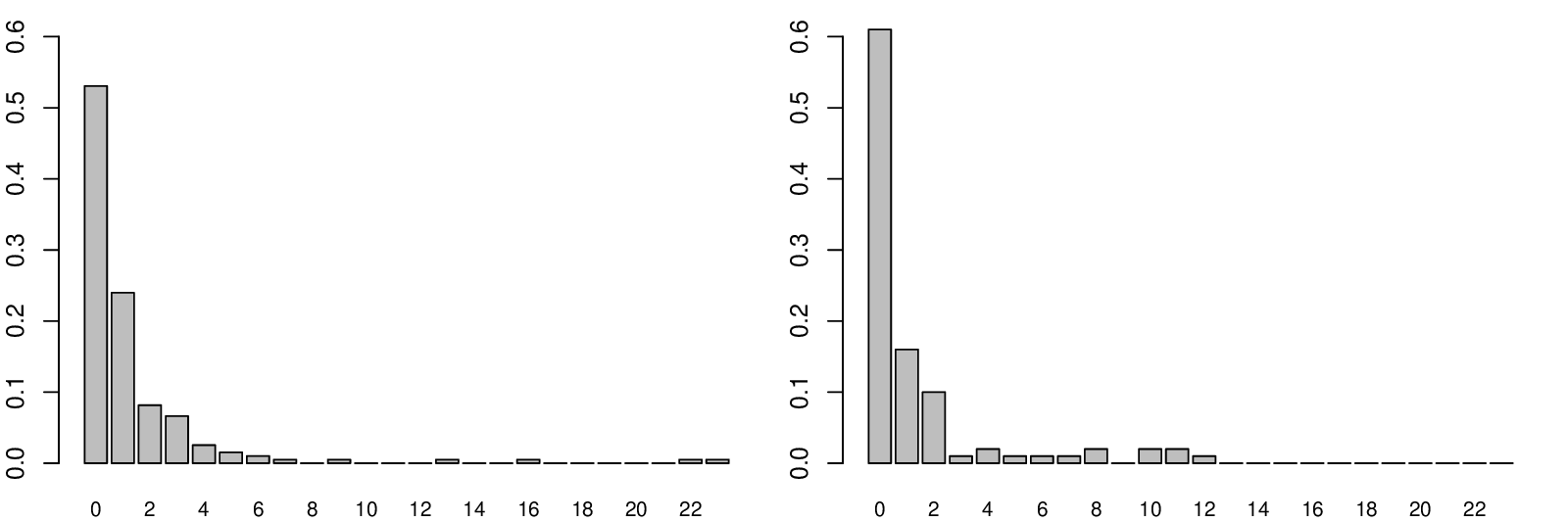}
	\end{center}
	\caption{Bar plots of the relative frequencies for the datasets \\ in Example \ref{ex2} \label{fig:ex2}}
\end{figure}

\medskip
{\bfseries Findings:} Under the assumption that both samples originate from the same distribution, one would expect the differences  
$$
d_m = \sup_{i\in \mathbb{N}_0} \left\{ \sum_{k=i}^{i+m} p_k \right\} -
\sup_{j\in \mathbb{N}_0} \left\{ \sum_{k=j}^{j+m} q_k \right\}
$$  
to be close to zero for all $m$, which is indeed the case. Since the sign of $d_m$ varies with $m$, the samples are not ordered with respect to $\leq_{\text{wd}}$ and, consequently, also not with respect to $\discoa$. Similarly, the differences in the dispersion measures between the two samples (see Table \ref{tab:cex2}) are small and also vary in sign:
\begin{table}[ht]
	{\smallskip\centering
		\begin{tabular}{c|ccccccc}
			& SD & MAD & GMD & IQR & $\widehat{\nu}_1$ & $\widehat{\nu}_2$ & $\widehat{\nu}_{\text{rob}}$ \\ \hline
			sample 1 & 2.94 & 1.53 & 1.99 & 1 & 0.65 & 1.61 & 0.79 \\
			sample 2 & 2.73 & 1.76 & 2.18 & 1 & 0.68 & 1.52 & 0.75
		\end{tabular}
		\par}
	\caption{Dispersion measures for the datasets in Example \ref{ex2} \label{tab:cex2} }
\end{table}

\subsection{Counts of intestinal parasites in European eels} \label{ex3}
The third example compares two samples of counts of intestinal parasites in European eels at the same location (River Rhine near Karlsruhe), but from different years. Sample 1 ($n=40$) was recorded in summer 1999, and sample 2 ($n=20$) in summer 2005 \citep{thielen2006}. The two samples come from quite different distributions, see \cite{klar:2010} for an explanation.
Figure \ref{fig:ex3} displays bar plots of the relative frequencies $p$ and $q$ for the two datasets. Table \ref{tab:ex3} presents the absolute frequencies $h_i (i=1,2)$ alongside the relative frequencies for both datasets.
\begin{table}[ht]
	\setlength{\tabcolsep}{2pt}
	\centering
	\small
	\begin{tabular}{rrrrrrrrrrrrrrrrrrrrrr} \hline
		$k$ & 0 & 1 & 2 & 3 & 4 & 5 & 6 & 7 & 8 & 9 & 10 & 11 & 12 & 14 & 27 & 37 & 39 & $\geq 40$ \\ 
		\hline
		$h_1$ & 14 & 3 & 5 & 4 & 2 & 1 & 1 & 1 & 0 & 0 & 1 & 1 & 0 & 0 & 1 & 1 & 1 & 4 \\ 
		$p$ & 0.35 & 0.07 & 0.12 & 0.10 & 0.05 & 0.02 & 0.02 & 0.02 & 0.00 & 0.00 & 0.02 & 0.02 & 0.00 & 0.00 & 0.02 & 0.02 & 0.02 & 0.10 \\ 
		$h_2$ & 1 & 0 & 2 & 1 & 3 & 2 & 1 & 0 & 1 & 2 & 0 & 2 & 2 & 1 & 1 & 0 & 1 & 0 \\ 
		$q$ & 0.05 & 0.00 & 0.10 & 0.05 & 0.15 & 0.10 & 0.05 & 0.00 & 0.05 & 0.10 & 0.00 & 0.10 & 0.10 & 0.05 & 0.05 & 0.00 & 0.05 & 0.00 \\ 
	\end{tabular}
	\caption{Absolute frequencies $h_1$ and $h_2$ and relative frequencies $p$ and $q$ for the datasets in Example \ref{ex3} \label{tab:ex3}}
\end{table}

\begin{figure}[ht]
	\begin{center}
		\includegraphics[width=\textwidth]{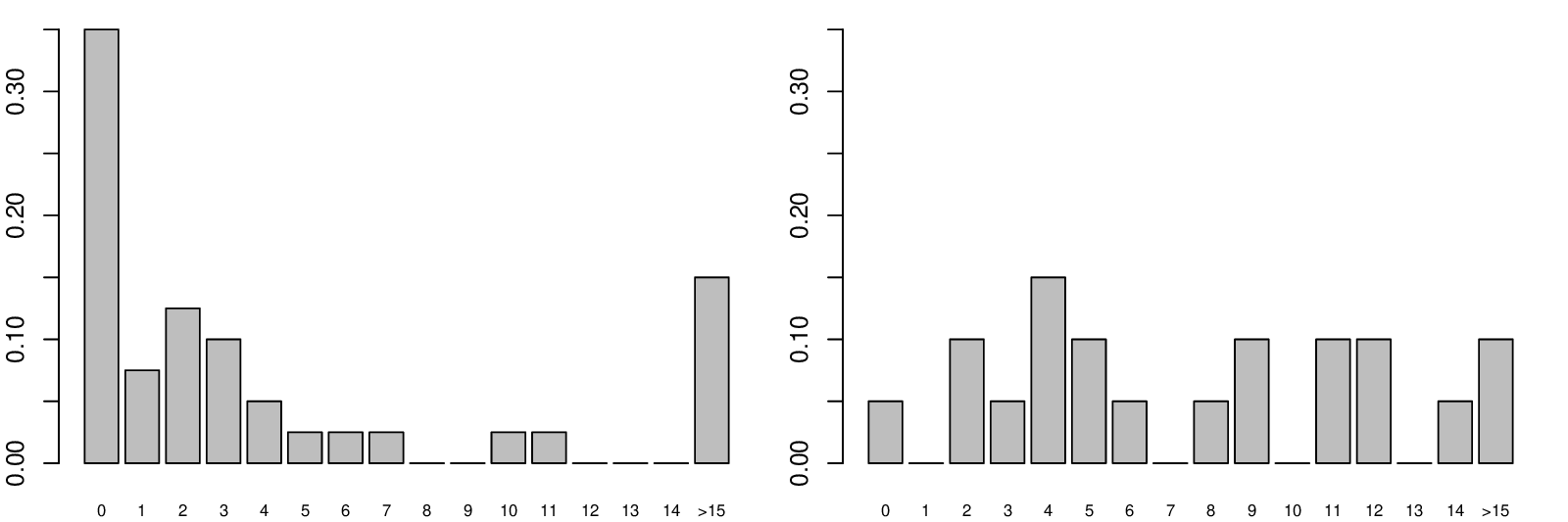}
	\end{center}
	\caption{Bar plots of the relative frequencies for the datasets \\ in Example \ref{ex3} \label{fig:ex3} }
\end{figure}

\smallskip
{\bfseries Findings:} Here, $d_m > 0$ for $m = 0, \ldots, 10$, while $d_m < 0$ for $m = 12, \ldots, 173$. Consequently, neither $X \leq_{\text{wd}} Y$ nor $Y \leq_{\text{wd}} X$ holds. Table \ref{tab:cex3} shows that the standard deviation, mean absolute deviation, and Gini mean difference are clearly larger for the first sample; the same applies to $\widehat{\nu}_1$ and $\widehat{\nu}_2$. In contrast, the differences in the interquartile ranges and the robust measure $\widehat{\nu}_{\text{rob}}$ are negative.

\begin{table}[ht]
	{\smallskip\centering
		\begin{tabular}{c|ccccccc}
			& SD & MAD & GMD & IQR & $\widehat{\nu}_1$ & $\widehat{\nu}_2$ & $\widehat{\nu}_{\text{rob}}$ \\ \hline
			sample 1 &35.83 &21.48 &25.68 & 6.25 & 7.60 &19.25 & 1.06 \\
			sample 2 & 9.19 & 6.06 & 8.93 & 7.25 & 4.07 & 6.23 & 1.23
		\end{tabular}
		\par}
	\caption{Dispersion measures for the datasets in Example \ref{ex3} \label{tab:cex3}}
\end{table}

\subsection{Amount of aggression attributed to film characters} \label{ex4}
To illustrate their test procedure, \cite{sigtuk1960} considered data from a study comparing the amount of aggression attributed to film characters by members of two populations, A and B, both with a sample size of $9$. Based on the test result, they concluded that population A is more spread out than population B.
Figure \ref{fig:ex4} displays bar plots of the relative frequencies $p$ and $q$ for the two datasets. Table \ref{tab:ex4} presents the absolute frequencies $h_i (i=1,2)$ alongside the relative frequencies for both datasets.
\begin{table}[ht]
	\setlength{\tabcolsep}{2pt}
	\centering
	\begin{tabular}{rrrrrrrrrrrrrrrr} \hline
		$k$ & 0 & 3 & 5 & 6 & 8 & 10 & 11 & 12 & 13 & 14 & 15 & 16 & 17 & 19 & 25 \\ \hline
		$h_1$ & 1 & 0 & 1 & 0 & 2 & 0 & 0 & 0 & 0 & 1 & 1 & 0 & 1 & 1 & 1 \\ 
		$p$ & 0.11 & 0.00 & 0.11 & 0.00 & 0.22 & 0.00 & 0.00 & 0.00 & 0.00 & 0.11 & 0.11 & 0.00 & 0.11 & 0.11 & 0.11 \\ \hline 
		$h_2$ & 0 & 1 & 0 & 1 & 0 & 2 & 1 & 1 & 2 & 0 & 0 & 1 & 0 & 0 & 0 \\ 
		$q$ & 0.00 & 0.11 & 0.00 & 0.11 & 0.00 & 0.22 & 0.11 & 0.11 & 0.22 & 0.00 & 0.00 & 0.11 & 0.00 & 0.00 & 0.00 \\ 
	\end{tabular}
	\caption{Absolute frequencies $h_1$ and $h_2$ and relative frequencies $p$ and $q$ for the datasets in Example \ref{ex4} \label{tab:ex4} }
\end{table}

\begin{figure}[ht]
	\begin{center}
		\includegraphics[width=\textwidth]{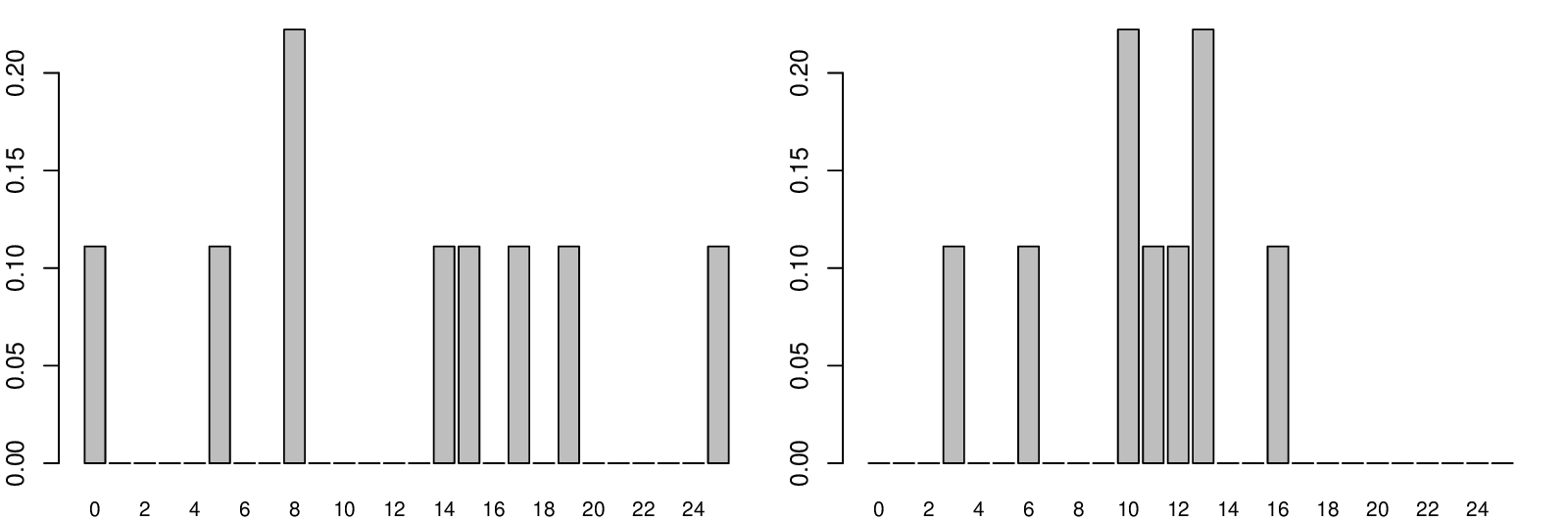}
	\end{center}
	\caption{Bar plots of the relative frequencies for the datasets \\ in Example \ref{ex4} \label{fig:ex4}}
\end{figure}

\smallskip
{\bfseries Findings:} As in the first example, $X >_{wd} Y$ holds. Therefore, we have $\nu(X) \geq \nu(Y)$ for the dispersion measures introduced in Section~\ref{sec5}. Once again, the two samples are not comparable with respect to $\discoa$, so no conclusion can be drawn regarding the other measures. However, Table \ref{tab:cex4} shows that, in this case as well, the inequality $\nu(X) > \nu(Y)$ holds for all measures considered:
\begin{table}[ht]
	{\smallskip\centering
		\begin{tabular}{c|ccccccc}
			& SD & MAD & GMD & IQR & $\widehat{\nu}_1$ & $\widehat{\nu}_2$ & $\widehat{\nu}_{\text{rob}}$ \\ \hline
			sample 1 & 7.75 & 6.30 & 9.28 & 9 & 4.78 & 6.28 & 1.23 \\
			sample 2 & 3.91 & 2.84 & 4.50 & 3 & 2.11 & 3.02 & 1.14
		\end{tabular}
		\par}
	\caption{Dispersion measures for the datasets in Example \ref{ex4} \label{tab:cex4} }
\end{table}

\smallskip
\begin{remark}
	\cite{ek-2024} present in Examples 2.7a) and 4.9a) two real-world datasets for which $X \discoa Y$ holds. Consequently, the classical dispersion measures SD, MAD, and GMD as well as the dispersion measures introduced in Section~\ref{sec5}, agree in their ranking: $\nu(X) \leq \nu(Y)$. In this case, we can confidently state that $Y$ exhibits higher variability than $X$. However, such a clear ordering is rarely encountered in real-world data.
	
	Subsections~\ref{ex1} and~\ref{ex4} provide examples where \( X >_{wd} Y \) holds. As a result, we have \( \nu(X) \geq \nu(Y) \) for the dispersion measures introduced in Section~\ref{sec5}. Thus, if we agree on using the order \( \leq_{\text{wd}} \) to assess dispersion, we can conclude that the variability of \( X \) is greater than that of \( Y \). Even though this ordering does not directly imply a relationship for classical dispersion measures, the examples show that these measures typically follow the same trend in such cases.
	
	In contrast, the example in subsection~\ref{ex3} demonstrates that relying on a specific dispersion measure can be misleading, even when the observed differences appear substantial - for instance, standard deviations of 36 and 9. Since the samples are not ordered with respect to \( \leq_{\text{wd}} \), neither sample should be deemed more variable than the other in a general sense.
\end{remark}

\section{Concluding remarks} \label{sec7}
This paper contributes to the ongoing development of dispersion concepts in statistics by proposing a weak dispersive order that is specifically designed for discrete distributions. This new order overcomes the limitations of existing approaches, particularly the classical dispersive order, which is overly restrictive in discrete settings due to its support-inclusion requirement. It also addresses the limitations of the discrete adaptation of this order introduced by \cite{ek-2024}, which is useful as a foundational concept but imposes strict conditions that are often difficult to verify. The proposed weak dispersive order is more flexible while preserving key interpretability features, and it enables meaningful comparisons of variability between a broader class of discrete distributions.

In addition to the theoretical framework, we introduced a family of variability measures derived from the Lévy concentration function. These measures satisfy the classical axioms of dispersion formulated by Bickel and Lehmann, offering robust and interpretable alternatives to standard dispersion measures in the discrete domain. Several empirical examples demonstrate the practical usefulness of the proposed ordering and associated measures, even in situations where classical dispersion measures may yield conflicting or unintuitive results.

There are several areas that could be explored in future research. For example, a detailed characterization of the equivalence classes under the weak dispersive order for non-unimodal distributions appears to be a challenging and interesting problem. Furthermore, the constructive form of the weak dispersive order - based on comparing probabilities over intervals of bounded length - suggests a natural pathway for extending the concept to multivariate discrete settings. One possibility would be to replace one-dimensional intervals with Euclidean balls or other metric-based neighborhoods, thereby enabling comparisons of spatially distributed count data, such as point patterns in ecology or materials science.

An additional point of interest lies in the relationship between the concepts of concentration and variability - where variability, in a narrower sense, refers to the spread around a central location. Both notions aim to capture aspects of dispersion; however, the examples suggest that standard measures of variability generally do not respect the weak dispersive order. By contrast, we have seen that concentration can serve as a lower bound for certain variability measures. This raises the following questions: Is there a more rigorous framework that formally clarifies the similarities and differences between these two concepts? For pairs of distributions ordered in variability according to weaker criteria than the dispersive order, such as the convex or dilation orders (see, e.g., \cite{shaked-sh-2006}), can we ensure that they are also ordered in concentration?

In summary, this work broadens the conceptual and methodological foundations for comparing variability in discrete data. It provides both theoretical insights and practical tools for statistical analysis in applied fields.

\bibliographystyle{plainnat}
\bibliography{bibliography}
\end{document}